\documentclass[letterpaper, 10 pt, conference]{ieeeconf}  
\usepackage{setspace}
\usepackage{amsmath,amsfonts}
\usepackage{array}
\usepackage[caption=false,font=normalsize,labelfont=normalfont,textfont=normalfont]{subfig}

\usepackage{caption}

\usepackage{textcomp}
\usepackage{stfloats}
\usepackage{url}
\usepackage{verbatim}
\usepackage{graphicx}
\usepackage{amssymb}
\usepackage{cite}
\usepackage{xcolor}

\usepackage{mathtools, nccmath}
\usepackage{pifont}
\newcommand{\xmark}{\ding{55}} 

\usepackage[ruled,vlined,linesnumbered]{algorithm2e}

\newtheorem{theorem}{Theorem}

\newtheorem{definition}[theorem]{Definition}

\newtheorem{lemma}[theorem]{Lemma}
\newtheorem{notation}[theorem]{Notation}

\newtheorem{assumption}[theorem]{Assumption}
\renewenvironment{proof}[1][Proof]{\noindent\textbf{#1.} }{\ \rule{0.5em}{0.5em}}

\usepackage{mathtools, nccmath}

\IEEEoverridecommandlockouts                              

\title{\LARGE \bf
Iterative Motion Planning in Multi-agent Systems with Opportunistic Communication under Disturbance}

\author{Neelanga Thelasingha$^{1}$, Agung Julius$^{1}$, James Humann$^{2}$, James Dotterweich$^{2}$
\thanks{*This work was supported by DEVCOM Army Research Laboratory under cooperative research agreement W911NF-21-2-0283 79054-CI-ARL}
\thanks{$^{1}$Neelanga Thelasingha and Agung Julius are with the Dept. of Electrical, Computer, and Systems Engineering, Rensselaer Polytechnic Institute, NY, USA.        {\tt\small thelan@rpi.edu, agung@ecse.rpi.edu}}%
\thanks{$^{2}$James Humann and James Dotterweich are with DEVCOM Army Research Laboratory, MD, USA. {\tt\small  james.d.humann.civ@army.mil}
}
\thanks{ \tt\small james.m.dotterweich.civ@army.mil}
}

\begin{document}

\maketitle
\thispagestyle{empty}
\pagestyle{empty}


\begin{abstract}
In complex multi-agent systems involving heterogeneous teams, uncertainty arises from numerous sources like environmental disturbances, model inaccuracies, and changing tasks. This causes planned trajectories to become infeasible, requiring replanning. Further, different communication architectures used in multi-agent systems give rise to asymmetric knowledge of planned trajectories across the agents. In such systems, replanning must be done in a communication-aware fashion. This paper establishes the conditions for synchronization and feasibility in epistemic planning scenarios introduced by opportunistic communication architectures. We also establish conditions on task satisfaction based on quantified recoverability of disturbances in an iterative planning scheme.  We further validate these theoretical results experimentally in a UAV--UGV task assignment problem.
\end{abstract}

\section{INTRODUCTION}

In multi-agent systems involving heterogeneous agents, uncertainty is introduced from numerous sources. In such instances, planning techniques must be robust and establish essential guarantees of performance \cite{chen2021scalable}. Mainly, disturbances occur due to external interactions, for example obstacles, turbulence  and modeling inaccuracies \cite{du2021cooperative,cheng2021fixed}. Further, most centralized communication methods expect agents to be constantly connected and in sync \cite{chriki:hal-02365868}. This is not feasible in real-world implementations. The routing command progression through the network decides the plan version on each agent. Therefore, partial knowledge exists where each agent's plan version depends on the synchronization. In this context, handling disturbances and replanning becomes a challenging problem due to asymmetrical knowledge. In this paper, we set up this problem in an opportunistic communication setting and derive requirements to enable plan synchronization when replanning. Then, we propose techniques to quantify disturbances, algorithms to handle them and establish requirements for task satisfaction. 

While disturbances cause the agents to replan, the task can change and lead to infeasibility \cite{mezgebe2020multi}. In literature, formulating planning tasks as Markov Decision Processes is a widely used approach \cite{Spaan2012}. However, increasing state dimensions increases the computational complexity \cite{taylor2019dynamic}. While data-driven methods are used to learn policies for task satisfaction, they are not sample-efficient \cite{yan2022optimal}. Meanwhile, receding horizon control \cite{grancharova2015uavs,tase} is used to periodically trigger replanning to compensate for the uncertainties \cite{YAO2016131}. Similar approaches are proposed in \cite{zhang2020distributed} using adaptive control. However, they do not consider the recoverability from disturbance. Thus, we identify that necessary conditions for recoverability must be derived to achieve task satisfaction.

In addition, a disturbance may lead to infeasibility due to the communication architecture. Some research suggests ad-hoc decentralized architectures, including widely used mesh networks\cite{Jawhar2017CommunicationAN} while multi-layer architectures operate in hierarchical interconnected sub-groups \cite{luan2021hierarchical,Sun2019ResearchOU}. The command routing strategy is also critical. Common topology-based static routing can be error-prone \cite{ChengHKV07} while reactive on-demand routing overcomes this by forming links as required. However, this gives rise to asymmetrical knowledge of plans across the system. In any architecture, the replan schedule depends on the speed of plan update propagation to agents without violating their known plan. Hence, each re-planning step must be communication-aware \cite{mardani2019communication}. To this end, formalisms such as "beliefs" have been used with epistemic logic \cite{bolander2011epistemic} to handle the problem of partial knowledge. In this paper, we analyze plan synchronization in an opportunistic communication setting. With the beliefs defined, we propose algorithms for planning and execution under opportunistic communication under disturbances. We also derive the necessary conditions to ensure task satisfaction.


We include a comparison of related work in multi-agent planning for task assignments in Table \ref{tab:comparison} highlighting our contributions. Specifically, we compare the disturbance handling strategy, the communication framework and synchronization approaches. In this paper, we explore synchronization challenges that arise in a multi-agent system due to opportunistic communication. We establish necessary conditions for plan synchronization and propose an algorithm for replanning. When the plans are executed by each agent, we propose a quantification of the recoverability of a disturbance and propose a recovery algorithm to handle the disturbances. We present the theoretical results on performance guarantees in disturbance handling and derive conditions for task satisfaction. Then, we present the results of experiments validating the theoretical results in a UAV-UGV task site assignment. 

\begin{table*}[t!]
    \centering
    \begin{tabular}{cccccc}
        \hline
       \multicolumn{1}{p{2cm}}{\centering Related Work} & \multicolumn{1}{p{2.1cm}}{\centering Disturbance Handling } & \multicolumn{1}{p{2.1cm}}{\centering Conditions on Satisfaction} & \multicolumn{1}{p{2.1cm}}{\centering Communication Topology} &   \multicolumn{1}{p{2.1cm}}{\centering Collaborative Tasks} & \multicolumn{1}{p{2.1cm}}{\centering Synchronization Mode}\\ \hline
        
        \cite{6224792,schillinger2018simultaneous} &  \xmark &  \xmark  & Not connected &  \xmark &  No sync \\
        \hline
        \cite{kantaros2020stylus,luo2022temporal,luo2021abstraction} &  \xmark &  \xmark  & Fully connected &  \checkmark &  Full \\
        \hline

    \cite{guo2015multi,ulusoy2013optimality,moarref2017decentralized} &  \xmark &  \xmark & Fully connected &  \checkmark &  Opportunistic \\
        \hline

        \cite{luan2021hierarchical,Khare2008AdhocNO} &  \xmark &  \xmark  & Dynamic  &  \checkmark &  Opportunistic \\
        \hline

    \cite{yan2022optimal,zhang2020distributed} &  \checkmark &  \xmark & Fully connected &  \checkmark &  Full \\
        \hline

    \cite{grancharova2015uavs} &  \checkmark &  \xmark  & Dynamic &  \checkmark &  Full \\
        \hline
        

 \textbf{Proposed approach} &  \checkmark &  \checkmark  & Hybrid &  \checkmark &  Opportunistic \\
        \hline

    \end{tabular}
    \caption{Comparison of Related Work : In communication topology, Fully connected -- Agents maintain a constant connection, Dynamic -- Agents dynamically trigger connections, Hybrid -- Both constant and dynamic connections. In synchronization mode, Full -- Agents are in sync at all times, Opportunistic -- Agents selectively sync triggered by events.}
    \label{tab:comparison}
\end{table*}

\section{Mathematical Framework}

In this section, we define the mathematical framework for analyzing the disturbances occurring in iterative multi-agent planning solutions for generalized task site assignment problems. We also state the necessary concepts for analyzing opportunistic communication.

\subsection{Generalized Task Site Assignment}

We represent the task site assignment for a multi-agent system as the planning problem in focus. It presents a set of task sites that need to be visited by an agent at some point in time, satisfying system constraints like unit movement and energy consumption. The task assignment is dynamic, and at each time step, the motion plan of the agents must satisfy the current assignment. A feasible solution is a plan that meets the requirement of visiting each task site. However, further optimization could improve a defined cost of the plan. Typically, these problems are solved by vehicle routing problem (VRP) solvers, though they are often challenging to scale. In this work, we look at an iterative solution architecture that uses opportunistic communication within agents. The synthesis and mathematical analysis of solutions require a proper modeling framework representing the dynamics in a heterogeneous multi-dimensional statespace.

\subsection{Preliminaries}

\begin{figure}[h!]
  \centering
  \includegraphics[width=0.5\textwidth]{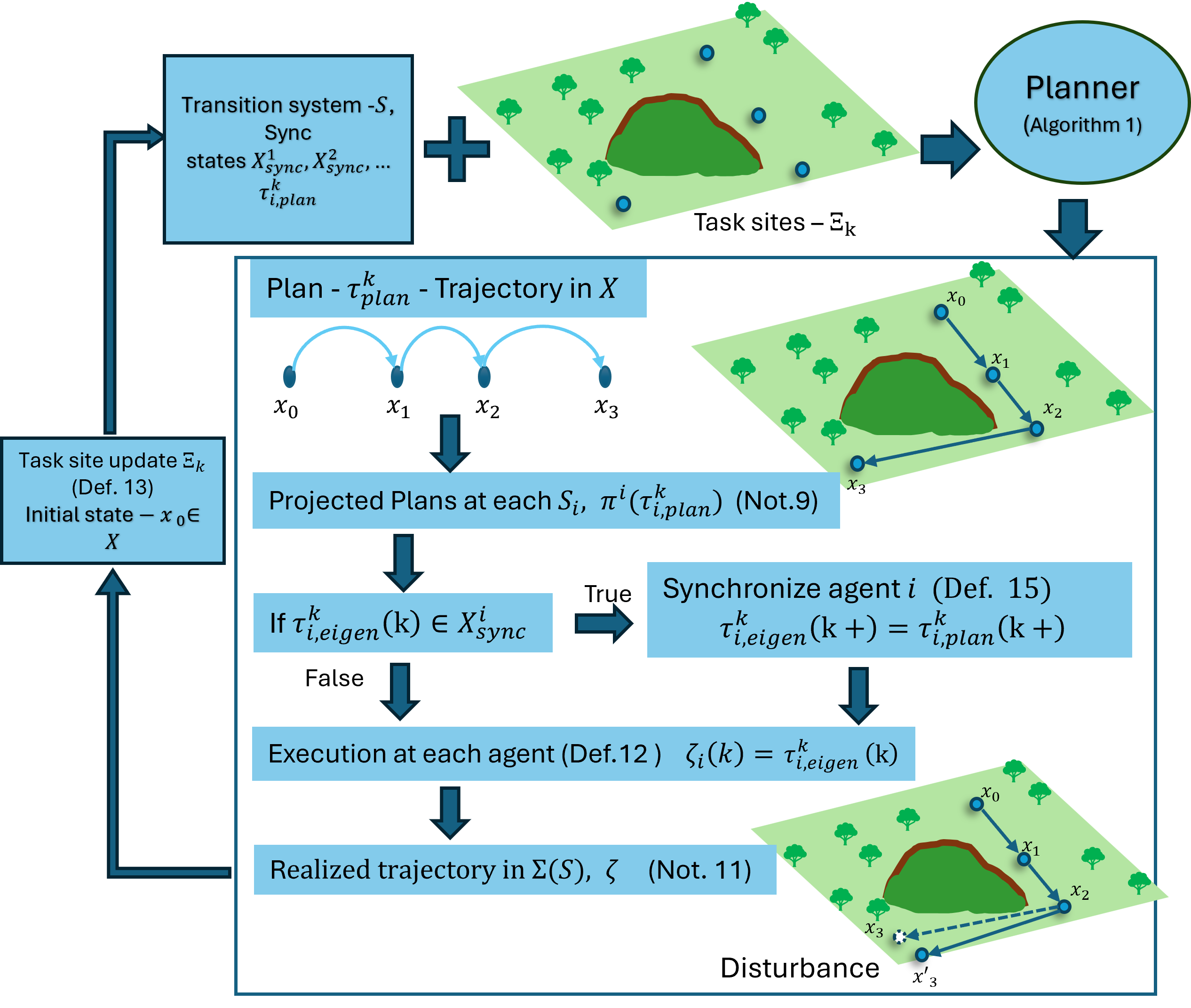} 
  \caption{Planning - Execution Scheme}
  \label{fig:scheme}
\end{figure}

We use transition systems as the mathematical model that describes a multi-agent system \cite{tabuada_book}. The multi-agent system dynamics can be represented as the transition system's state evolution through time. First, we formally define the agent transition system as follows and multi-agent transition systems as a composition of several agent transition systems. 


\begin{notation}
[Agent Transition System]
An agent transition system is defined for an agent $i$ as a pair $S_i=(X_i,T_i)$, where $X_i$ is the state space (or set of states), $T_i\subset X_i \times X_i$ is the set of transitions and $t=(x,x^{\prime})$ denotes a transition from a state $x\in X_i$ to the state $x^{\prime}\in X_i$. The physical interpretation of a transition is the evolution of the state of the agent through time in a given duration. We assume transitions take constant time throughout this paper.
\end{notation}

\begin{notation}
[Multi-agent Transition System]
Given a set of agent transition systems $S_i=(X_i,T_i)$ for agents $i=1,\dotsc,N_A$, a transition system $S=(X,T)=S_1\times \dotsc \times  S_{N_A}$ where the set of states $X=X_1\times\dotsc \times X_{N_A}$, and the set of transitions $T\triangleq \{ (x,x^\prime)\ |\ x=[x_1,\dotsc,x_{N_A}]\in X,\  x^\prime=[x^\prime_1,\dotsc,x^\prime_{N_A}]\in X,\ \forall i\ (x_i,x^\prime_i)\in T_i \}$, is a multi-agent transition system.
\end{notation}

The evolution of the state within a transition system through a sequence of transitions is defined as a trajectory.

\begin{definition}
[Trajectory] For any positive integer $N$ and a transition system $S=(X,T)$, a sequence of states $\tau
:=x_{0}\cdots x_{N}$, 
is said to be a trajectory of $S$ iff  $t=(x_{k},x_{k+1}) \in T$ for all $ k\in\{0,\cdots,N-1\}$.
$N$ is the length of trajectory $\tau$. The set of all trajectories of $S$ is denoted by $\Sigma(S)$. 
\end{definition}
\begin{notation}
    Any state $\hat{x}\in X$ or transition $t \in T$, is said to be in a trajectory $\tau =x_{0}\cdots x_{N}$ if $\hat{x}=x_{i}$ or $t=(x_{i},x_{i+1})$ respectively for some $i \in \{0,\cdots,N\}$. It is denoted by $\hat{x}\in \tau$ or $ t \in \tau$ respectively. Any state $x_k\in \tau$ at $k^{th}$ time step is denoted as $x_k=\tau(k)$. Also, we denote the sequence of states after any time step $k$ as $\tau(k+)$.
\end{notation}

The trajectory for the multi-agent system consists of trajectories for each agent as defined by the projected trajectories. 

\begin{definition}[Projection]
    Let \( S= (X,T) \) be the multi-agent transition system and \( S_i= (X_i,T_i) \) the transition system for agent \( i \). We define a projection operator \( \pi_i \colon X \to X_i \) which extracts the state of agent \( i \) from a multi-agent state, such that,
\[
\pi_i(x) = x^i, \quad \text{for } x = [x^1, x^2, \dots, x^{N_A}] \in X.
\]
We lift this notation to trajectories by defining a projection operator \( \Pi_i \colon \Sigma(S) \to \Sigma(S_i) \). For a multi-agent trajectory \( \tau = x_k x_{k+1} \dotsc \in \Sigma(S) \), the $i^{th}$ agent trajectory is given by the  projection such that,
    \[
    \Pi_i(\tau) = \pi_i(x_k) \pi_i(x_{k+1}) \dotsc = x^i_k x^i_{k+1} \dotsc \in \Sigma(S_i).
    \]
\end{definition}

\begin{definition} [Task Site Assignment] For a multi-agent transition system $S=(X,T)$ of $N_A$ agents, a task site assignment is a collection of state classes, $\Xi=\{\hat{X_1},\dotsc,\hat{X_{N_x}}\}$, where each $\hat{X_j}\subset X$. Here $N_x$ is the cardinality of $\Xi$. An agent $i$ is said to satisfy a $\hat{X_j}\subset X$ at a time step $k$ if the state of $i$, $x_i=\pi_i(x_k)$ such that $x_k \in \hat{X_j}$. A trajectory $\tau$ satisfies the task site assignment $\Xi$, if for all $j=1,\dotsc, N_x$, there exists some agent $i$ that satisfies $\hat{X_j} \in \Xi$ at some time stamp $k$ such that $x_i=\pi_i(\tau(k))$ where $\tau(k) \in \hat{X_j}$. It is denoted as $\tau \models \Xi$.
\end{definition}




The goal of the planning problem is to find a task-satisfying trajectory for execution. This is done by the Task Site Assignment (TSA) planner for the multi-agent system. Within the TSA planner, we assume the availability of a solver as in \cite{tase} that can generate trajectories to satisfy the task site assignment.

\begin{assumption} [TSA Planner]\label{def:oracle} 
Given a multi-agent transition system $S$, an initial state $x_0\in X$, and a feasible task site assignment $\Xi$, it is assumed that there is a solver $\mathcal{O}$ that generates task satisfying trajectories. $\mathcal{O}$ manifests a trajectory $\tau \models \Xi$ such that $x_0 = \tau(0)$. The TSA planner calls $\mathcal{O}$ to generate plans and communicate to agents.
\end{assumption}

As the plan updates over time, TSA planner and agents have varying knowledge of the plan. Therefore, we define the following notations to identify the plan versions.

\begin{notation}
     For a multi-agent transition system $S$ at time step $k$, $\tau^{k}_{plan}(k+)\in \Sigma(S)$ denotes the plan for all agents in $S$ as known by the TSA planner at $k$. 
\end{notation}

\begin{notation}
    For an agent transition system $S_i$ at a time step $k$, the trajectory $\tau^{k}_{i,eigen} \in \Sigma(S_i)$ denotes the future plan for agent $i$ as known by agent $i$ at time step $k$. Similarly, the trajectory $\tau^{k}_{i,plan} \in \Sigma(S_i)$ denotes the future plan for agent $i$ as known by the TSA planner at time step $k$, that is, $\tau^{k}_{i,plan}$  is the belief of $\tau^{k}_{i,eigen} $ by the planner.
\end{notation}

The eigen trajectories are realized in the system following the planning-execution scheme in Fig. \ref{fig:scheme}. Next, we define the realized trajectory for an agent transition system and extend it to multi-agent trasnition systems as follows.

\begin{definition}[Realized Trajectory]
    For an agent transition system, $S_i=(X_i,T_i)$, the realized trajectory is a trajectory followed by the agent from time step $0$ to any future time step $k>0$. We denote a realized trajectory of agent $i$ as $\zeta_i= x^i_{0}x^i_{1}\cdots x^i_{k}$ where states $x^i_j \in X_i$ for all $j$, transition $(x^i_{j},x^i_{j+1})\in T_i$ for all $j$ and  $\zeta_i \in \Sigma(S_i)$. The state at any time step $j$ is denoted by $\zeta_i(j)$. Also, we denote the sequence of states after any time step $k$ as $\zeta(k+)$.
\end{definition}

\begin{notation}\label{not:realized_t}
    For a multi-agent system $S=(X,T)$, agent transition systems $S_i=(X_i,T_i)$ and realized trajectories $\zeta_i$ for $i=1,\dotsc,N_A$, we denote the realized trajectory of $S$ as $\zeta=x_0 \dotsc x_k$ where each $x_j=[x^1_j,\dotsc,x^{N_A}_j]\in X$ and $x^i_j=\zeta_i(j)$ for all $j=0,\dotsc,k$ and $i=1,\dotsc,N_A$.
\end{notation}

The realization of a trajectory is done by executing the plan by each agent. We define perfect execution assuming the absence of any disturbance.

\begin{definition}[Perfect Execution]\label{def:executn}
    Let $S_i$ be an agent transition system and a trajectory $\tau^{k}_{i,eigen}(k+)\in \Sigma(S_i)$ be the plan known to agent $i$ at a time step $k$. The realized trajectory is related to the time-varying plan at time step $k$ by $\zeta_i(k)=\tau^{k}_{i,eigen}(k)$.
\end{definition}

At the execution the initial state is updated as the final state of the system for the next planning iteration at the planner, as shown in Fig. \ref{fig:scheme}. At a plan-execute time step $k$, let $x^k_0$ denote the initial state, $\Xi^k$ denote the current task site assignment, and $\tau^k$ denote the future trajectory(plan) generated by a solver for $x^k_0$ and $\Xi^k$ such that $x^k_0 = \zeta(k-1)$ where  $\zeta$ is the realized trajectory at time step $k-1$. 

\begin{definition}[Task site update]
    Given a task site assignment $\Xi^k$ at time step $k$, $\Xi^k$ is updated to  $\Xi^{k+1}$ at time step $k+1$ such that as $\Xi^{k+1}=(\Xi^{k}\setminus \Bar{\Xi}^{k+1}_{-}) \cup \Bar{\Xi}^{k+1}_{+}$, where $\Bar{\Xi}^{k+1}_{-}$ denotes the set of state classes that are removed from $\Xi^{k}$ and $\Bar{\Xi}^{k+1}_{+}$ denotes the set of state classes that are added to $\Xi^{k}$.
\end{definition}
    
After executing the time step $k$, the initial state and task site assignment are updated to $x^{k+1}_0=\zeta(k)$ and $\Xi^{k+1}$, which are then used to generate the plan at $k+1$, $\tau^{k+1}$. 

\section{Requirements for Plan Synchronization}

\subsection{Opportunistic Communication}

The challenge of handling disturbances is exacerbated by the communication protocol used by the multi-agent system. We explore this in the context of opportunistic communication. By definition, a multi-agent transition system is constructed as a composition of at least one single-agent transition system. When the planner generates plans for the multi-agent system, it includes plans for all the agents. However, due to opportunistic communication, generated plans are not instantly communicated.


The projected trajectories for agents $\Pi^i(\tau^{k}_{plan})$, are not instantly communicated. This happens opportunistically through synchronization between agents and planner. Given a multi-agent transition system $S$ of $N_A$ agents, at any time step $k$, synchronization is the communication of future plans $\tau^{k}_{plan}(k+)$ to agent $i$, that is setting $\tau^{k}_{i,eigen}(k+)=\Pi^i(\tau^{k}_{plan}(k+))$ where $i=1,\dotsc,N_A$. This is limited to certain states defined by the system and task definition. Thus, we define synchronization states as follows.


    

\begin{definition} [Synchronization states]
    For a multi-agent transition system $S$ and an agent $i$ where $i=1,\dotsc,N_A$, we denote $X^i_{sync} \subset X$ as the synchronization state for the $i^{th}$ agent where plan synchronization for $i$ is possible.
\end{definition}

Synchronization is possible if the system state is in $X^i_{sync}$. This restricts the agents' opportunities for plan updates and creates mismatches in the knowledge of the agent and planner on future plans.
\begin{definition} [Opportunistic Synchronization]\label{def:syncact}
    For $S$, a multi-agent transition system, $S_i$, agent transition systems where $i=1,\dotsc, N_A$, and synchronization states for each agent $i$ $X^i_{sync} \subset X$, let the plan in $S$ be $\tau^{k}_{plan}(k+)$. Then synchronization for agent $i$ occurs as below.
    \begin{equation}
    \tau^{k}_{i,eigen}(k+)= \begin{cases}
\Pi^i(\tau^{k}_{plan}(k+)),&\text{if } x_k \in X^i_{sync},\\
\tau^{k-1}_{i,eigen}(k+),&\text{if } x_k \notin X^i_{sync},
\end{cases}
\end{equation}
and, 
 \begin{equation}
    \tau^{k}_{i,plan}(k+)= \begin{cases}
\tau^{k}_{i,eigen}(k+),&\text{if } x_k \in X^i_{sync},\\
\tau^{k-1}_{i,plan}(k+),&\text{if } x_k \notin X^i_{sync},
\end{cases}
\end{equation}
where $x_k=\tau^{k}_{plan}(k)$ is the multi-agent system state at $k$. At $k=0$, we assume full synchronization.
\end{definition}

Ensuring synchronization is critical in disturbance handling under opportunistic communication scenarios, as shown in Fig. \ref{fig:improvem}. Thus, a planning and synchronization framework that supports epistemic constructs is essential.

\begin{figure}[h]
\centering
\includegraphics[width=0.5\textwidth]{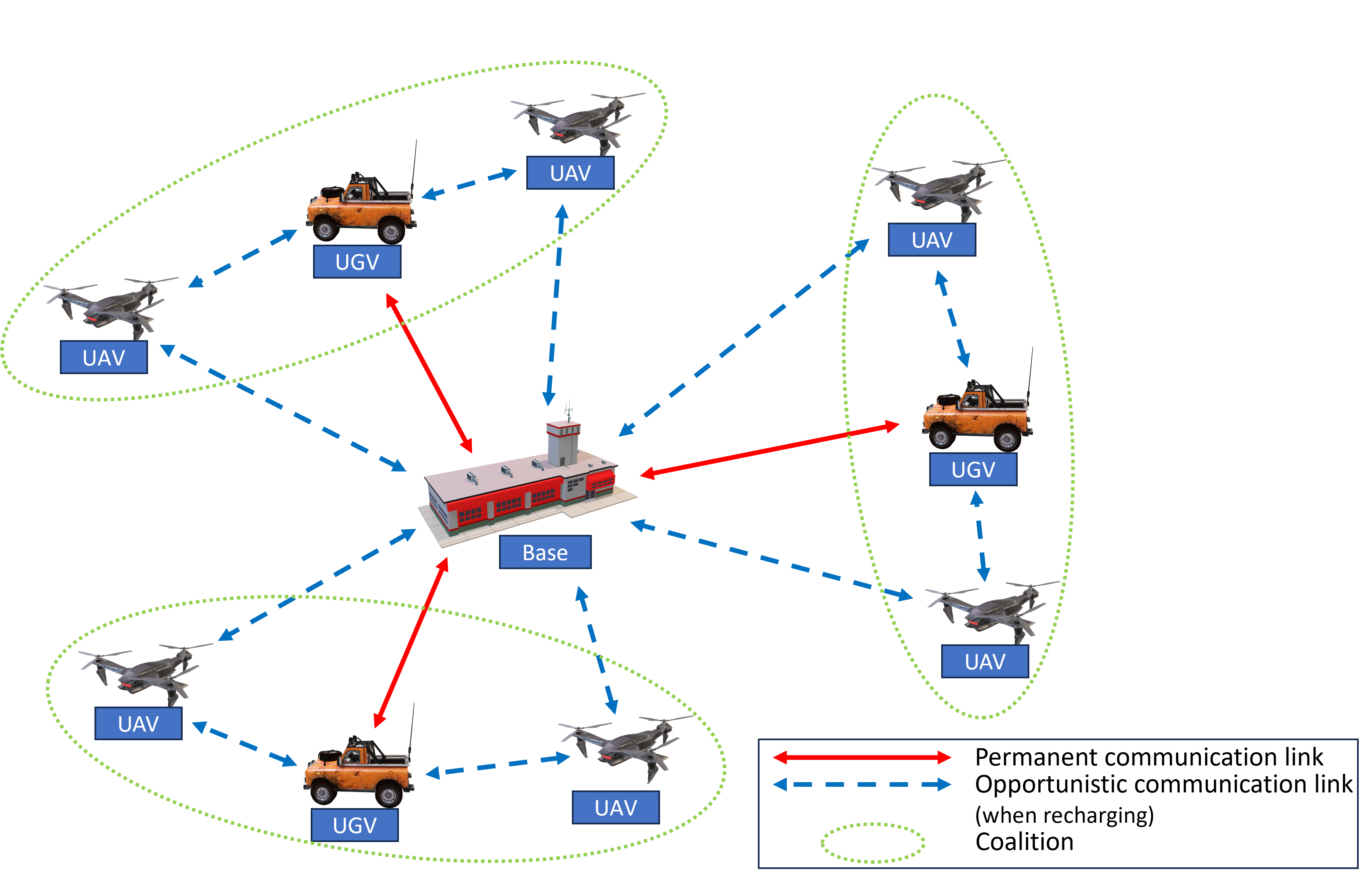}
\caption{Communication architecture, UAV-UGV coalitions, permanent and opportunistic links. UGVs and bases are always in sync. UAVs establish links opportunistically at synchronization states with bases and UGVs in their coalitions.}
\label{fig:improvem}
\end{figure}

\subsection{Synchronization Requirements}

Due to opportunistic communication, the plans for the multi-agent system are communicated only at sync states. Therefore, the planner must account for opportunistic synchronization. To this end, we state Algorithm \ref{alg:planning} for planning with synchronization, which ensures that the plans generated for the multi-agent system at each planning point satisfy the synchronization conditions. The plans generated from Algorithm \ref{alg:planning}, when executed, satisfy the task site assignment. Algorithm \ref{alg:planning} is called at each step to generate a plan for the multi-agent system for the current task site assignment $\Xi^k$. We state the following assumption on the existence of a constraint solver for algorithm \ref{alg:planning}.


\begin{assumption}[Constraint Solver]\label{def:consolve}  
Let \( S = (X, T) \) be a multi-agent transition system, and \( \Xi_k \) a feasible task site assignment at time $k$ and $\tau^{k-1}_{i,plan}(\kappa_i)$, the plan at time $k-1$. We assume the existence of a constraint solver, $\mathcal{O}(S, x^k, \Xi^k, \Phi, \Psi)$ that generates $ \tau^k_{plan} \in \Sigma(S)$ where, \( x^k \in X \) is the system state at time step \( k \), \( \Phi = \{(i, \kappa_i)\} \) specifies agent \( i \)'s trajectory points such that  \( \tau^k_{plan}(\kappa_i) = \tau^{k-1}_{i,plan}(\kappa_i)\), and  \( \Psi = \{(i, k^*_i)\} \) specifies the next synchronization states \( \tau^k_{plan}(k^*_i) \in X^i_{sync} \).
\end{assumption}

 The planner calls the constraint solver to generate the plan, but not all agents are in synchronization. For the agents that are not in synchronization, the planners' version of agents' eigen plan ($\tau^{k-1}_{i,plan}$) is used to define constraints on the new plan, such that the segment of agent $i$s' plan until the next expected synchronization is kept unchanged in the new plan $\tau^{k}_{plan}$. Thus, Algorithm  \ref{alg:planning} ensures synchronization for the absent agents when the new plan is generated.

\begin{algorithm}
\RestyleAlgo{ruled}
\caption{Planning with Opportunistic Synchronization}\label{alg:planning}
\KwData{ $S=(X,T)$- Multi-agent transition system,  $N_A$- Number of agents,  $\Xi^k $ -  Task site assignments at time step $k$, $x^k\in X$ - System state at $k$, $\tau^{k-1}_{plan}$ - Current plan for multi-agent system,  $\mathcal{O}$ - Solver, For all $i \in \{1,\dotsc, N_A\}$ $\tau^{k-1}_{i,plan}$ - Current plan for agent $i$ as known by planner,  $X^i_{sync}$ - Synchronization states for $i$.}
\KwResult{$\tau^{k}_{plan}$}
$\Phi \gets \emptyset \subseteq \mathbb{R}^2$\;
$\Psi \gets \emptyset \subseteq \mathbb{R}^2$\;
 \For{$i=1,\dotsc,N_A$}{
 \If{$x^k \notin X^i_{sync}$}
 {
 $k^*_i \gets \min \{k_i>k \ |\ \tau^{k-1}_{plan}(k_i) \in  X^i_{sync} \}$\;
 Add $(i,k^*_i)$ to $\Psi$\;
 $\forall \kappa_i \in \{k,\dotsc, k^*_i\} $, add $(i,\kappa_i)$ to $\Phi$ \;
 }
 }
 Call Solver $\mathcal{O}$ with $x^k$ to generate a $\tau^{k}_{plan} \models \Xi^k $ s. t constraint $\pi^{i}(\tau^{k}_{plan}(\kappa_i))=\tau^{k-1}_{i,plan}(\kappa_i))$ is satisfied $\forall (i,\kappa_i)\in\Phi$ and constraint $\tau^{k}_{plan}(k^*_{i}))\in X^{i}_{sync}$  is satisfied $\forall (i,k^*_{i})\in \Psi$ \;
 \Return $\tau^{k}_{plan}$\;
\end{algorithm}

We state the following theorem on task satisfaction.
\begin{theorem}[Task Satisfaction] \label{theo:satisfaction}
    For a multi-agent transition system $S$ of $N_A$ agents, let $\Xi_k$ be the task site assignment at any time step $k$. Given a plan $\tau^k_{plan}$ that is generated from Algorithm \ref{alg:planning}, and assuming perfect execution, the realized trajectory after $k$, $\zeta(k+)$, satisfies $\Xi_k$. That is, $\zeta(k+) \models \Xi_k$, if for all time steps $\gamma > k$, $\Bar{\Xi}^{\gamma}_{-}=\Bar{\Xi}^{\gamma}_{+}=\emptyset$.
\end{theorem}

\begin{proof}
    If for all time steps $\gamma > k$, $\Bar{\Xi}^{\gamma}_{-}=\Bar{\Xi}^{\gamma}_{+}=\emptyset$, then $\Xi_{\gamma}=\Xi_{\gamma-1}=\Xi_k$. Thus, for all $\gamma>k$, $\tau^{\gamma}_{plan}\models \Xi_k$. 
    Therefore, if each agent $i=1,\dotsc,N_A$ executes, $\tau^{\gamma}_{i,plan}=\Pi^i(\tau^{\gamma}_{plan})$, then the multi agent system satisfies $\Xi_k$.
    However, agent $i$ always executes $\tau^{\gamma}_{i,eigen}$. Thus $\tau^{\gamma}_{i,eigen}$ must be updated to $\tau^{\gamma}_{i,plan}$ for task satisfaction. The next time step (after $k$) where this synchronization is possible for agent $i$ is $k^*_i = \inf \{k_i>k\ |\ \tau^{k}_{plan}(k_i) \in  X^i_{sync} \}$. Therefore, from Definition \ref{def:syncact}, for $\gamma \geq k^*_i$, the plan is synchronized for agent $i$ at $k^*_i$, that is $\tau^{k^*_i}_{i,eigen}( k^*_i+)=\Pi^i(\tau^{k^*_i}_{plan}( k^*_i+))$.
    From Algorithm \ref{alg:planning},  $\pi^i(\tau^{\gamma_i}_{plan}(\gamma_i))=\tau^{\gamma}_{i,eigen}(\gamma_i)$ for all $ \gamma_i \in \{k ,\dotsc, k^*_i\}$ for each agent $i$.  Therefore, for all time steps after $k$, $\tau^{k}_{i,eigen}(k+)=\Pi^i(\tau^{k}_{plan}(k+))$. From Definition \ref{def:executn}, we have $\zeta_i(\gamma)=\tau^{k}_{i,eigen}(\gamma)$ for the realized trajectory $\zeta_i$ from perfect execution of the plan for agent $i$ at any time step $\gamma>k$. Thus, $\zeta_i(k+)=\Pi^i(\tau^{k}_{plan}(k+))$. From Notation \ref{not:realized_t}, the realized trajectory is $\zeta(k+)=x_k x_{k+1} \dotsc $ where each $x_{\gamma}=[x^1_{\gamma},x^2_{\gamma},\dotsc,x^{N_A}_{\gamma}]\in X$ and $x^i_{\gamma}=\zeta_i({\gamma})=\pi^i(\tau^{k}_{plan}(\gamma))$ for all ${\gamma}\geq k$ and $i=1,2,\dotsc,N_A$. 
    Therefore, $\zeta(k+)=\tau^{k}_{plan}(k+)$. Since, $\tau^k_{plan}(k+) \models \Xi_k$, we have $\zeta(k+) \models \Xi_k$.\end{proof}

\section{Bounded Performance Degradation}

\subsection{Disturbance Formulation}

At each time step $k$, agent $i$ executes $\tau^k_{i,eigen}$, realizing the next step of the trajectory. Disturbances occur when realizing a trajectory and perfect execution is not achieved.



\begin{definition}
[Disturbance]
Given a multi-agent transition system $S$ of $N_A$ agents and a task site assignment $\Xi^k$ at time step $k$, let $\tau^k_{plan}(k+) \models \Xi^k$ be a trajectory in $\Sigma(S)$ generated from Algorithm \ref{alg:planning}. At the execution, let the realized trajectory of any agent $i$ be $\zeta_i$. A mismatch such that $\zeta_i(k) \neq \pi_i(\tau^k_{plan}(k))$ is said to be a disturbance $\Bar{x_k}$ at time step $k$ for agent $i$ where $\Bar{x_k}= \zeta_i(k)$.
\end{definition}

When a disturbance occurs at any time step $k$, the multi-agent system state ventures away from the plan $\tau^k_{plan}$. Therefore, the trajectory must be corrected for task satisfaction. The requirements and conditions for trajectory correction need to be identified. Also, guarantees must be made on the performance degradation resulting from the disturbance. 

\subsection{Disturbance Recovery}

Now, we analyze the performance degradation under disturbance. The goal is to establish performance guarantees using disturbance recoverability and provide bounds on possible performance degradation. We first define the following reachability property on the states of transition systems.

\begin{definition}
[n-step Backward Reachable Set] For an agent transition system $S_i=(X_i,T_i)$, the n-step backward reachable set for a state $x_k \in X_i$ is defined as  $Reach_n(x_k)= \{ x_{k-n} \in X_i\ |\ \exists \tau \in \Sigma^n (S_i) \text{ s.t. } \tau(k-n)=x_{k-n} \text{ and } \ \tau(k)=x_k \}$
\end{definition}

Thus, the n-step backward reachable set for any state $x \in X_i$ is the set of states $Reach_n(x)\in X_i$ from which a sequence of $n$ transitions exist to reach $x$. If a disturbance is contained within an n-step backward reachable set of a future state of the plan, the agent can return to the planned trajectory within n steps. We further extend this reachability formulation to analyze the recoverability of a disturbance. When a disturbance happens at execution, the perfect execution assumption no longer holds, and the previously planned trajectory is no longer followed. Thus, the disturbed agent must recover to the next possible state in the plan where synchronization is possible and update the planner. To this end, we state the following definition of n-step recoverability.





\begin{definition}  
[n-step Recoverability]\label{def:nsteprec}
At time step $k$, let $\tau^{k}_{i,eigen}(k+)$ denote the plan for agent $i$ as known by agent $i$. Let $\Bar{x_k}=\zeta_i(k)\neq \tau^{k}_{i,eigen}(k)$ be a disturbance where $\zeta_i(k)$ is the realized trajectory. If there exists an $n>0$, such that  $\Bar{x_k} \in Reach_n(\tau^{k}_{i,eigen}(k+n))$, then the disturbance $\Bar{x_k}$ is said to be $n$-step recoverable. 
\end{definition}

This allows the approach described in Algorithm \ref{alg:distrubance} for handling disturbance. In Algorithm \ref{alg:distrubance}, each step of the trajectory is executed, and if a disturbance occurs and assuming it is n-step recoverable as in Definition \ref{def:nsteprec}, it is handled by calling the Solver to plan a recovery trajectory.

\begin{algorithm}
\RestyleAlgo{ruled}
\caption{Execution under disturbance at Agent}\label{alg:distrubance}
\KwData{$S_i$ - Agent transition system for the $i^{th}$ Agent,  $\tau^{k}_{i,eigen}$ - Plan for agent $i$ as known by $i$ at time step $k$}
Execute time step $k$ to yield state $x^i_k=\zeta_i(k)$\;
\If{$x^i_k \neq \tau^{k}_{i,eigen}(k)$}{
    $n^*= \min \{n \ |\ x^i_k \in Reach_n(\tau^{k}_{i,eigen}(k+n)) \}$\;
    Generate a trajectory $\tau^{k+1}_{i,eigen}(k+)$ such that  $\tau^{k+1}_{i,eigen}(k)=x^i_k$,  $\tau^{k+1}_{i,eigen}(k+n^*)=\tau^{k}_{i,eigen}(k+n^*)$, and $\tau^{k+1}_{i,eigen}((k+n^*)+)=\tau^{k}_{i,eigen}((k+n^*)+)$
}

\end{algorithm}

\subsection{Planning under Disturbance}

When the perfect execution assumption no longer holds for the system under disturbance, conditions for Theorem \ref{theo:satisfaction} are not satisfied. Further, plans from Algorithm \ref{alg:planning} become infeasible. Hence, we state Algorithm \ref{alg:planning_dist} for planning under opportunistic synchronization with disturbance.

\begin{algorithm}
\RestyleAlgo{ruled}
\caption{Planning with Opportunistic Synchronization and Disturbance}\label{alg:planning_dist}
\KwData{ $S=(X,T)$- Multi-agent transition system,  $N_A$- Number of agents,  $\Xi^k $ -  Task site assignments at time step $k$, $x^k\in X$ - System state at $k$, $\tau^{k-1}_{plan}$ - Current plan for multi-agent system,  $\mathcal{O}$ - Solver, For all $i \in \{1,\dotsc, N_A\}$ $\tau^{k-1}_{i,plan}$ - Current plan for agent $i$ as known by planner,  $X^i_{sync}$ - Synchronization states for $i$.}
\KwResult{$\tau^{k}_{plan}$}
$\Phi \gets \emptyset \subseteq \mathbb{R}^2$\;
$\Psi \gets \emptyset \subseteq \mathbb{R}^2$\;
 \For{$i=1,\dotsc,N_A$}{
 \If{$x^k \notin X^i_{sync}$}
 {
$\mathcal{K} \gets \{k_i>k \ |\ \tau^{k-1}_{plan}(k_i) \in  X^i_{sync} \}$\;
  $\forall \kappa_i \in \mathcal{K}$, add $(i,\kappa_i)$ to $\Psi$ \;
  $\forall \kappa_i \in \{k,\dotsc,  \inf \mathcal{K}\} $, add $(i,\kappa_i)$ to $\Phi$ \;
 }
 }
 Call Solver $\mathcal{O}$ with $x^k$ to generate a $\tau^{k}_{plan} \models \Xi^k $ s. t constraint $\tau^{k}_{plan}(\kappa_i) \in X^{i}_{sync}$ is satisfied $\forall (i,\kappa_i)\in\Psi$ and  $\pi^{i}(\tau^{k}_{plan}(\kappa_i))=\tau^{k-1}_{i,plan}(\kappa_i))$ is satisfied $\forall (i,\kappa_i)\in\Phi$ \;
 \Return $\tau^{k}_{plan}$\;
\end{algorithm}

We have established the algorithm for execution, handling recoverable disturbances. However, when the agent recovers in $n^*$ steps, the originally planned trajectory will not be followed for at most $n^*$ steps. Plan synchronization at the next synchronization state must be ensured for task satisfaction.
\begin{lemma}[Synchronization under Disturbance]\label{lemma:sync}
     For a multi-agent transition system $S$ of $N_A$ agents, let a plan at time step $k$ be $\tau^k_{plan}$ which is generated from Algorithm \ref{alg:planning_dist}. For an agent $i = 1,\dotsc, N_A$, let $k_{i}^{*}=\min \{k_i > k | \tau^k_{plan}(k_i) \in X^i_{sync} \}$. Assuming execution under Algorithm \ref{alg:distrubance}, agent $i$ will achieve synchronization at time $k_{i}^{*}$, if any disturbance that occurs for any agent, at time $\kappa \in [k,k_{i}^{*}]$, is $(k_{i}^{*}-\kappa)$-step recoverable for all $\kappa$.
\end{lemma}

\begin{proof}
    If any disturbance at time step $\kappa \in [k,k_{i}^{*}]$, is $(k_{i}^{*}-\kappa)$-step recoverable for all $\kappa$, then  $\pi^i(\tau^k_{plan}(k_{i}^{*}))$ is always reachable for all agents $i=1,\dotsc, N_A$. Also, from Algorithm \ref{alg:planning_dist}, $\tau^k_{plan}(k_{i}^{*}) \in X^i_{sync}$. Execution under Algorithm \ref{alg:distrubance} always finds a path towards the closest reachable state at future time step $n^*= \min \{n \ |\ x^i_k \in Reach_n(\tau^{k}_{i,eigen}(k+n)) \}$\, where $n^* \geq k_{i}^{*}$ due to the $(k_{i}^{*}-\kappa)$-step recoverability. Therefore, executed recovery trajectory always include, $\pi^i(\tau^k_{plan}(k_{i}^{*}))$, thus agent $i$ always synchronizes at $k_{i}^{*}$.
\end{proof}

Then, we state the following lemma to establish requirements for task satisfaction under disturbance.

\begin{lemma}[Satisfaction under Disturbance]\label{lemma:task}
For a multi-agent transition system $S$ of $N_A$ agents and a task site assignment $\Xi_k$, let plan at time step $k$ be $\tau^k_{plan} \models \Xi_k$ which is generated from Algorithm \ref{alg:planning_dist}. For an agent $i = 1,\dotsc, N_A$, let $ k_{i}^{*}=\min \{ k_i > k | \tau^k_{plan}(k_i) \in X^i_{sync} \}$  be the next synchronization point. Also, let $k^{\prime}_{i}= \max \{k_i \leq k_{i}^{*} | \exists \hat{X} \in \Xi_k \ \text{where}\ \tau^k_{plan}(k_i) \in \hat{X} \}$. 
Further, let the corresponding task site be $\hat{X}^{\prime}$ where $\tau^k_{plan}(k^{\prime}_{max}) \in \hat{X}^{\prime}$. Assuming execution under Algorithm \ref{alg:distrubance}, agent $i$ will satisfy task site $\hat{X}^{\prime}$ at time $k^{\prime}_i$, if any disturbance that occurs at time $\kappa \in [k,k^{\prime}_{max}]$ for agent $i$ is $(k^{\prime}_{i} - \kappa)$-step recoverable for all $\kappa$.
\end{lemma}


\begin{proof}
     To satisfy the task site $\hat{X}^{\prime}$ by an agent $i$ at time $k_i$, the state $\pi^{i}(\tau^k_{plan}(k^{\prime}_{i}))$ must be included in the eigen trajectory of agent $i$. If the disturbance at time $\kappa \in [k,k^{\prime}_{i}]$, is $(k^{\prime}_{i}-\kappa)$-step recoverable for all $\kappa$, then the state $\pi^{i}(\tau^k_{plan}(k^{\prime}_{i}))$ is always reachable for the agent $i$.
     Execution under Algorithm \ref{alg:distrubance} always finds a path to the closest reachable state $n^*$ steps ahead where $n^*= \min \{n \ |\ x^i_k \in Reach_n(\tau^{k}_{i,eigen}(\kappa+n)) \}$.  Here, $n^* \leq (k^{\prime}_{i}-\kappa)$ from $(k^{\prime}_{i}-\kappa)$-step recoverability. Therefore, if disturbed, the executed trajectory always includes the state $\pi^{i}(\tau^k_{plan}(k^{\prime}_{i}))$, thus agent $i$ always satisfies the task site $\hat{X}^{\prime}$ at time $k^{\prime}_{i}$.
\end{proof}

From the above lemmas \ref{lemma:sync} and \ref{lemma:task}, we state the following theorem on task satisfaction. 


\begin{theorem}[Satisfaction under Disturbance] \label{theo:Distrubed_satisfaction}
For a multi-agent transition system $S$ of $N_A$ agents, let $\Xi_k$ be the task assignment at time step $k$. Given a plan $\tau^k_{plan}$ from Algorithm \ref{alg:planning_dist}, let $k^{\prime}_{max} = \max_{i \in \{1,\ldots,N_A\}} k^{\prime}_i$ and $i_{max} = \arg \max_{i \in \{1,\ldots,N_A\}} k^{\prime}_i$, where $k^{\prime}_i=\max \{k_i \leq k{i}^{*} | \exists \hat{X} \in \Xi_k \text{ s.t. } \tau^{k}_{plan}(k_i) \in \hat{X} \}$ and $k_{i}^{*}=\min \{k_i > k | \tau^{k}_{plan}(k_i) \in X^i_{sync} \}$. Further, let the task site be $\hat{X}^{\prime}$ where $\tau^{k}_{plan}(k^{\prime}_{max}) \in \hat{X}^{\prime}$. Assuming execution under Algorithm \ref{alg:distrubance}, realized trajectory after $k$, $\zeta(k+)$ will satisfy task assignment $\Xi_k$, if following conditions hold;\\
1. For all $\gamma>k$, $\Bar{\Xi}^{\gamma}{+}=\emptyset$\\
2. All disturbances that occur at time $\gamma \in [k,k_{i}^{*}]$ that affect agent $i_{max}$ is $(k_{i}^{*} - \gamma)$-step recoverable for all $\gamma$.\\
3. All disturbances that occur at time $\gamma \in [k,k^{\prime}_{max}]$ that affect agent $i_{max}$ are $(k^{\prime}_{max} - \gamma)$-step recoverable for all $\gamma$.
\end{theorem}

\begin{proof}
    From condition 1, if for all time steps $\gamma > k$, $=\Bar{\Xi}^{\gamma}_{+}=\emptyset$, then $\Xi_{\gamma} \subseteq \Xi_{\gamma-1} \subseteq \Xi_k$. Thus, for all $\gamma>k$, $\tau^{\gamma}_{plan}\models \Xi_k$. 
    Therefore, if each agent $i=1,\dotsc,N_A$ executes, $\tau^{\gamma}_{i,plan}=\Pi^i(\tau^{\gamma}_{plan})$, then $\Xi_k$ is satisfied.
    Given condition 2's recoverability, from Lemma \ref{lemma:sync}, synchronization is possible for agent $i_{max}$ at all time steps after $k$ at $k_{i_{max}}^{*}=\min \{k_{i_{max}} > k | \tau^k_{plan}(k_{i_{max}}) \in X^i_{sync} \}$. Therefore, all agents achieve synchronization as planned. 
    Given condition 3's disturbance recoverability, from Lemma \ref{lemma:task} at least 1 task site will be satisfied at timestep  $k^{\prime}_{max}$ by agent $i_{max}$ for all time steps after $k$. Thus, the realized trajectory $\zeta(k+)$ satisfies at least one task site in each synchronization step. If the cardinality of the task assignment, $\Xi_k$ is $N_{\Xi}$, then $\Xi_k$ is satisfied in at most $N_{\Xi}$ synchronization steps given that any disturbance satisfies the conditions 2 and 3. Therefore, we have $\zeta(k+) \models \Xi_k$.\end{proof}

\begin{figure}[h!]
\centering
\includegraphics[width=0.4\textwidth]{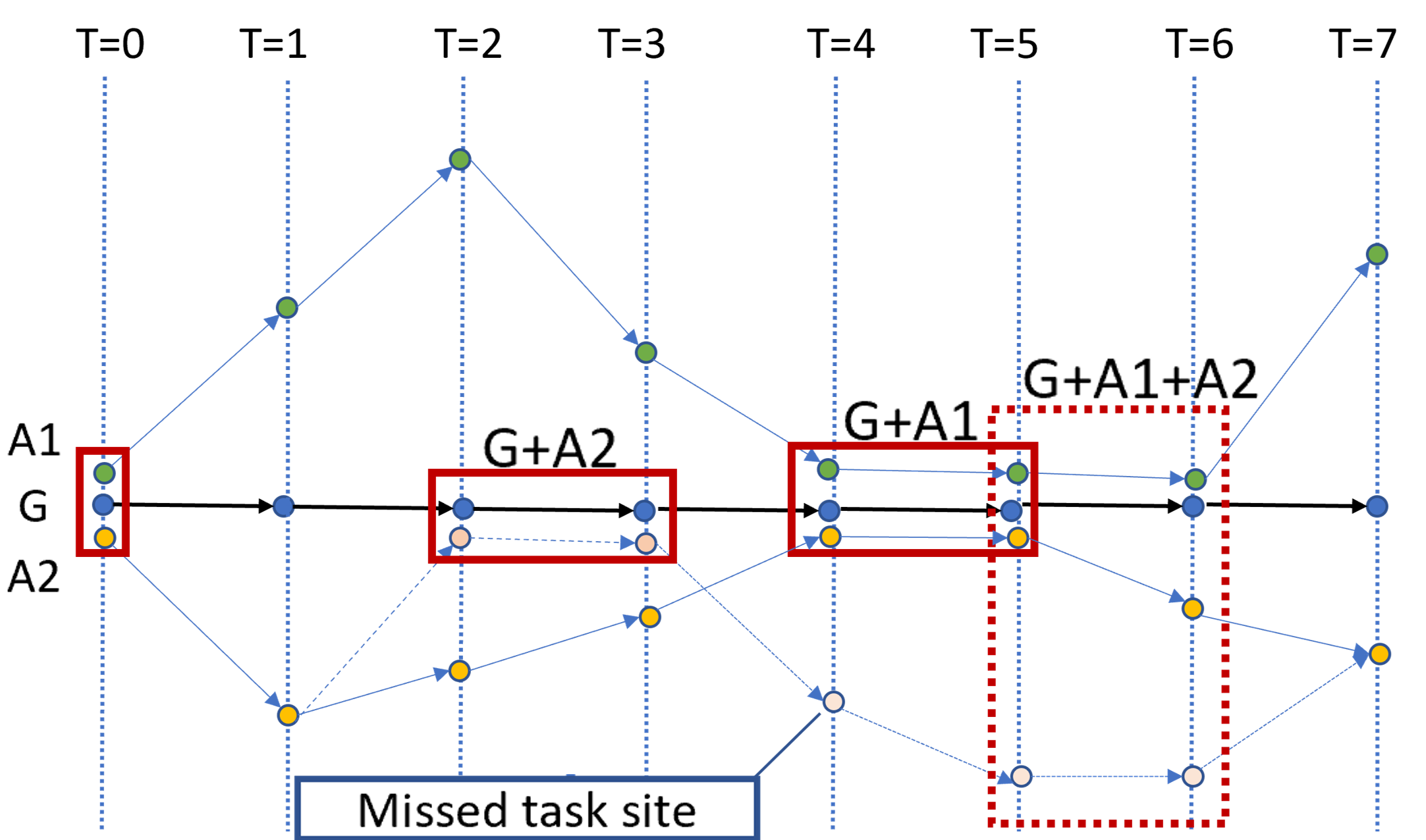}
\caption{Demonstration of plan execution with synchronization}
\label{fig:demo}
\end{figure}

We demonstrate the synchronization and replanning procedure in Fig. \ref{fig:demo}. Two UAVS $A_1$, $A_2$ and UGV $G$ start in sync, but at $T=1$, no UAV is in sync. At $T=2$, $A_2$ is to be synchronized, but it misses the sync state due to a disturbance. Therefore, it plans a recovery trajectory. All agents synchronize and replan at the next sync state at $T=4$. This illustrates how Algorithms 1, 2, and  3 are applied.

Now, we have proved that the proposed recovery algorithm achieves satisfaction, conditioned on disturbance recoverability. Next, we experimentally validate the proposed methods.

\section{EXPERIMENTAL RESULTS}

\begin{figure*}[t!]
\centering
\includegraphics[width=0.75\textwidth]{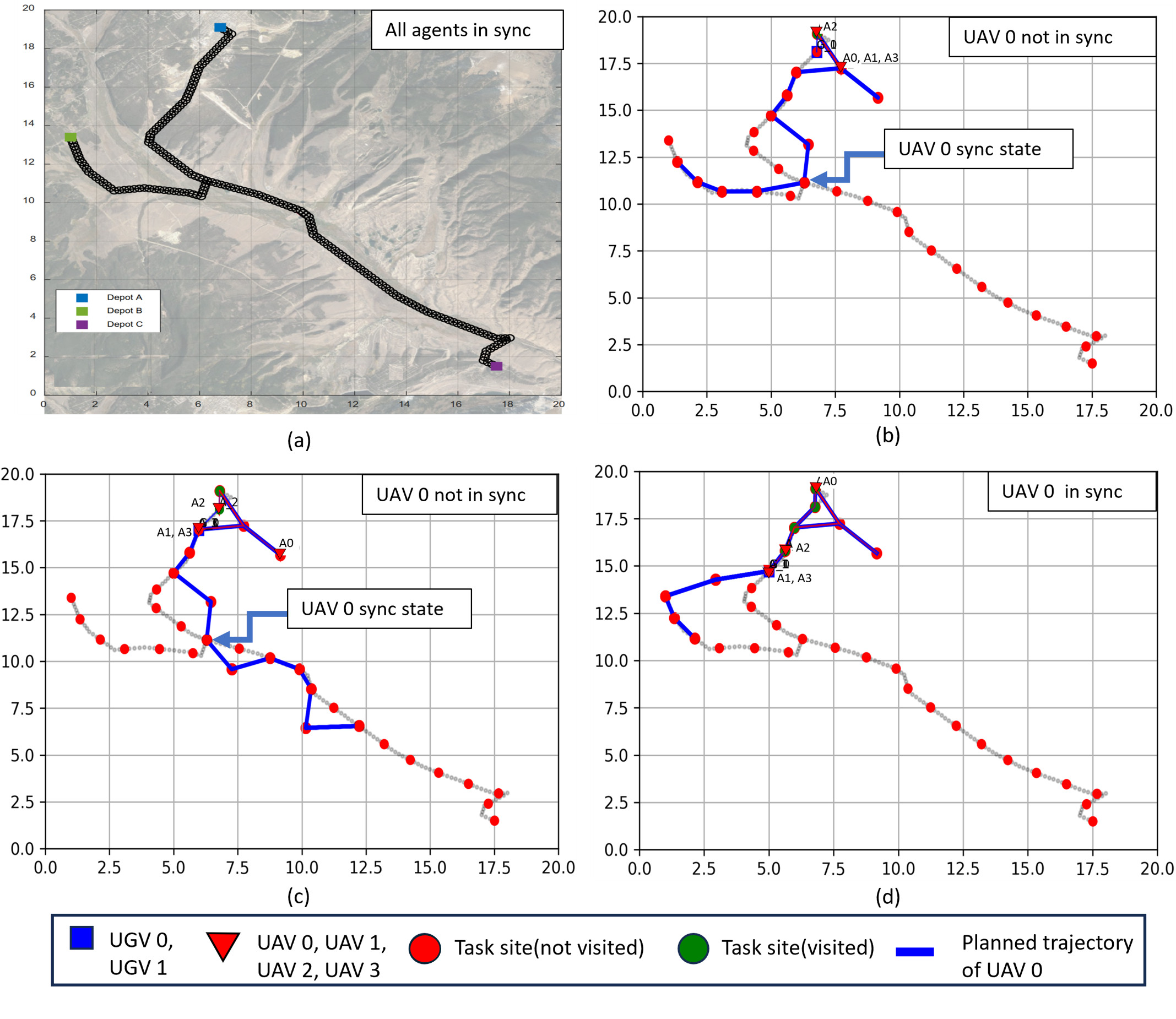}
\caption{Visualized execution trajectory with synchronization: (a) Task assignment is visiting the road. All agents are at dept A and in sync (b) UAV 0's plan is shown. As it is not in sync, the plan only changes beyond the sync state (c) UAV 0's plan changed beyond the sync state (d) A disturbance occurred and a recovery trajectory was planned for UAV 0.}
\label{fig:revisions}
\end{figure*}

\subsection{Implementation Setting}

\color{black}

We implement the proposed approach in the routing problem from \cite{tase}. Consider a 2D map as in Fig. \ref{fig:revisions}.a with a road and three depots. The task is to monitor the road with a team of UAVs and UGVs. UAVs synchronize with the planner when recharging. In the implementation, 2 UGVS and 4 UAVs are used to demonstrate the proposed algorithms. Each UGV forms a coalition with 2 UAVs and opportunistically synchronizes as shown in Fig. \ref{fig:improvem}. UGVs are restricted to the road and have an energy capacity of $B^{\mathrm{UGV}}=25.01 \,\mathrm{MJ}$. The power consumption is $1.05(464.8v_g +356.3)\mathrm{W}$ where $v_g$ is the speed. In contrast, the UAVs can travel freely within the map, but have limited energy of $B^{\mathrm{UAV}}=287.7\ \mathrm{kJ}$. The power consumption curve is $1.05(0.0461v_a^3-0.5834v_a^2- 1.8761v_a+ 229.6)\mathrm{W}$ where $v_a$ is the speed. UAV recharge rate depends on the current energy level. We denote the coordinate of $j^{th}$ agent as $(p^j_{x},p^j_{y})\in G $, energy level as $e_j\in [0,B^j_{max}]$ where $B^j_{max}$ is the battery capacity, and agent status $f_j$ encoding type and docking status. Then, for $N_A$ agents, the state of agent $j$ is $x_j=[p^j_{x},p^j_{y},e_j, f_j]$. The transition system  at time $\gamma$ is defined as $S=(X,T)$ such that $X=\{x_1, x_2,..., x_{N_A},\gamma \}$ and $T=\{(x,x^\prime)\ |\ x,x^\prime \in X\}.$ A transition is defined as an evolution of state in constant time under the dynamics. The continuous spatial, energy, and time state increase the planning complexity as it involves searching for solutions in multi-dimensional statespace.
\color{black}
\begin{figure}[h!]
\centering
\includegraphics[width=0.45\textwidth]{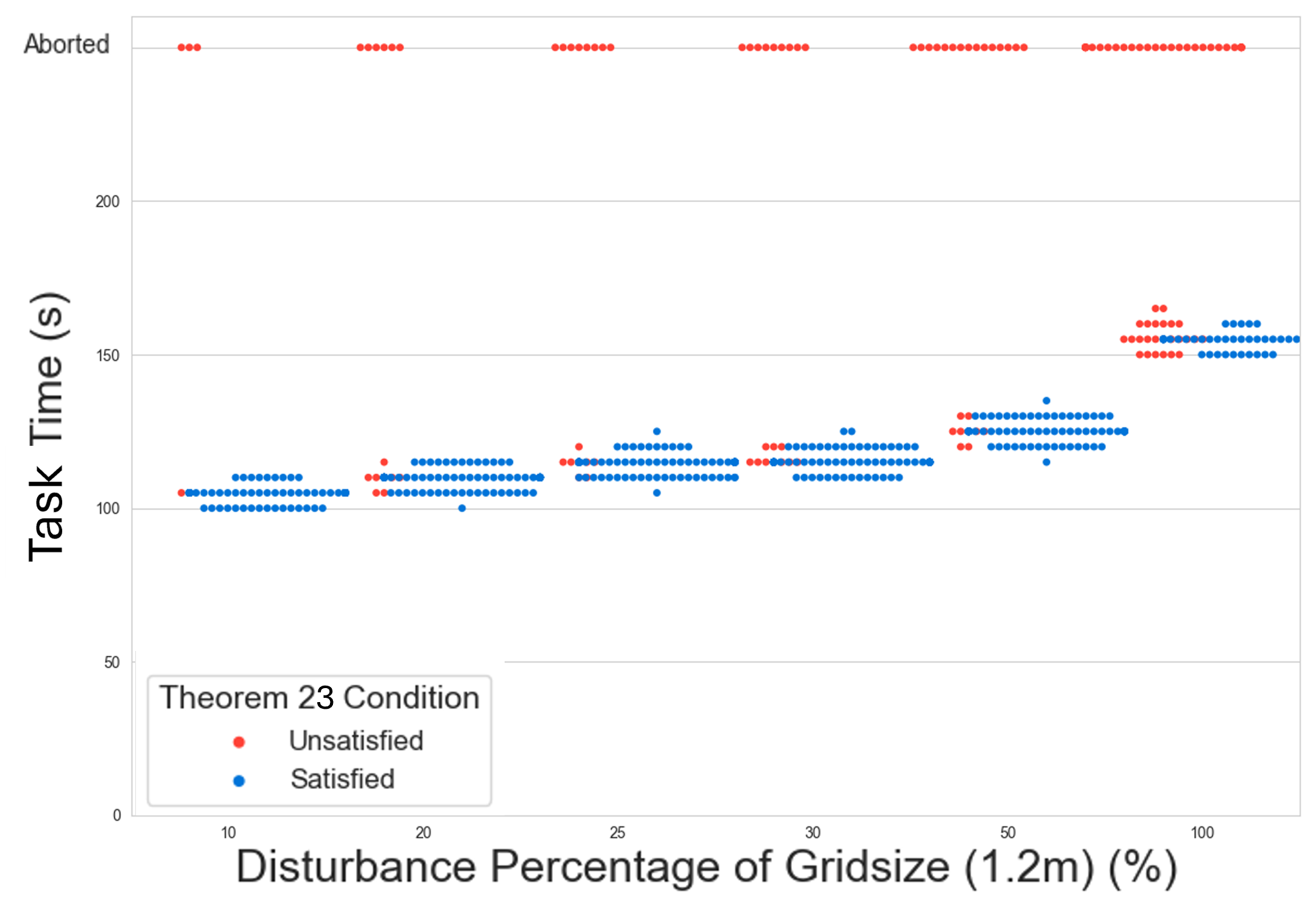}
\caption{Disturbance vs. Task time. Aborted instances and Theorem \ref{theo:Distrubed_satisfaction} condition satisfaction are also shown.}
\label{fig:distrub}
\end{figure}

\subsection{Results: Disturbances vs. Task time}
We observe how the task completion time (Task time) when applying increasing levels of disturbances. Fig. \ref{fig:distrub} confirms the theoretical results in Theorem \ref{theo:Distrubed_satisfaction}. When the conditions in Theorem \ref{theo:Distrubed_satisfaction} are satisfied,s the system is guaranteed to recover from the disturbance, and the realized trajectory satisfies the task assignment. For the applied disturbances, the recovery time depends on the severity, and in some instances, the task becomes infeasible and aborted.

\color{black}

\subsection{Results: Visualized plan execution and synchronization}

We present a result visualization in Fig. \ref{fig:revisions}. At each step, the trajectory is adjusted. In sub-figure (b) and (c), we see the plan is only changed beyond the synchronization point. Following the theoretical derivations, the plan before synchronization is kept constant as it cannot be updated before the agent is in sync. However, at step (d), a disturbance has occurred and UAV 0 has revised trajectory through depot A and it is currently at depot A. Therefore, it is in synchronization, so the entire plan is updated for UAV 0.


\color{black}
\section{CONCLUSIONS}

In this paper, we establish theoretical guarantees on the disturbance handling strategy in an iterative planning method for a multi-agent system. When disturbances occur when motion planning methods are deployed, such guarantees are essential to ensure performance. We introduce a quantification of the disturbances based on their recoverability. Then, we prove that when the disturbances are recoverable, conditions can be derived to ensure task satisfaction. This allows for establishing guarantees for performance, given the recoverability metrics of the disturbances.
We also analyze the plan synchronization problem when opportunistic communication architecture is used with the iterative planning strategy. The hybrid communication architecture uses a mesh network and a hierarchical implementation to distribute routing commands. At each replanning step, when a new plan is found for the system, synchronization should occur to transmit the plan to each agent. To this end, we derived the requirements for plan synchronization in an opportunistic communication framework for the iterative planning strategy. 
By addressing these issues, we contribute to the theoretical analysis of iterative planning methods in multi-agent systems, laying the groundwork for more robust and efficient operations in real-world scenarios. We believe future research could explore the performance analysis of multi-agent systems in dynamic environments.










\bibliographystyle{IEEEtran}
\bibliography{bibdata}

\end{document}